\documentclass[12pt, draftclsnofoot, onecolumn]{IEEEtran}
\usepackage[utf8]{inputenc}
\usepackage[english]{babel}
\usepackage{cite}
\usepackage{setspace}
\usepackage{graphicx}
\usepackage{amsmath}
\usepackage{algorithm}
\usepackage{algpseudocode}

\usepackage{algcompatible}
\usepackage{array}
\usepackage{amsthm}
\usepackage[caption=false,font=normalsize,labelfont=sf,textfont=sf]{subfig}
\usepackage{fixltx2e}
\usepackage{txfonts}
\usepackage{stfloats}
\usepackage{url}
\newtheorem{theorem}{\textbf{Theorem}} 
 
\newtheorem{lemma}{\textbf{Lemma}}

\newtheorem{definition}{\textbf{Definition}}

\usepackage{lipsum}  
\usepackage{color}

\allowdisplaybreaks

\begin{document}
\title{Minimum Length Scheduling for Multi-cell Full Duplex Wireless Powered Communication Networks
}
\author{Muhammad Shahid~Iqbal, Yalcin Sadi,~\IEEEmembership{Member,~IEEE}
        and~Sinem~Coleri,~\IEEEmembership{Member,~IEEE}
\thanks{Muhammad Shahid~Iqbal and Sinem Coleri are with the department of Electrical and Electronics Engineering, Koc University, Istanbul, Turkey, email: $\lbrace$miqbal16, scoleri$\rbrace$@ku.edu.tr. Yalcin Sadi is with the department of Electrical and Electronics Engineering, Kadir Has University, Istanbul, Turkey, email: yalcin.sadi@khas.edu.tr.

This work is supported by Scientific and Technological Research Council of Turkey Grant $\#$117E241.

}}
\maketitle
\begin{abstract}
Wireless powered communication networks (WPCNs) are crucial in achieving perpetual lifetime for future wireless communication networks. Practical WPCNs cover large area and large number of sensors, necessitating multi-cell deployment. In this paper, we consider a multi-cell full-duplex WPCN, in which multiple hybrid access points (HAPs) transfer energy to a set of users that concurrently transmit information to their respective HAP. We investigate a novel minimum length scheduling problem to determine the optimal power control, and scheduling for constant and continuous rate models, while considering concurrent transmission of users, energy causality, maximum transmit power and traffic demand constraints. The formulated optimization problems are shown to be non-convex and combinatorial in nature, thus, difficult to solve for the global optimum. As a solution strategy, first, we propose optimal polynomial time algorithms for the power control problem considering constant and continuous rate models based on the evaluation of Perron-Frobenius conditions and usage of bisection method, respectively. Then, the proposed optimal power control solutions are used to determine the optimal transmission time for a subset of users that will be scheduled by the scheduling algorithms. For the constant rate scheduling problem, we propose a heuristic algorithm that aims to maximize the number of concurrently transmitting users by maximizing the allowable interference on each user without violating the signal-to-noise-ratio (SNR) requirements. For the continuous rate scheduling problem, we define a penalty function representing the advantage of concurrent transmission over individual transmission of the users. Following the optimality analysis of the penalty metric and demonstration of the equivalence between schedule length minimization and minimization of the sum of penalties, we propose a heuristic algorithm based on the allocation of the concurrently transmitting users with the goal of minimizing the sum penalties over the schedule. Through extensive simulations, we demonstrate that the proposed algorithm outperforms the successive transmission and concurrent transmission of randomly selected users for different HAP transmit powers, network densities and network size. 
\end{abstract}
\begin{IEEEkeywords}
Energy harvesting, wireless powered communication networks, full duplex networks, power control, scheduling, multi-cell network.
\end{IEEEkeywords}
\IEEEpeerreviewmaketitle
\section{Introduction} \label{sec:intro}
Thanks to recent advances and standardization in communication technology for Internet of things (IoT) and machine type communication (MTC), including Bluetooth $5$, WiFi Halow, narrow band IoT and wakeup radios, wireless sensor networks are now part of our daily life \cite{MTC-IoTstandards}. According to the recent Ericsson mobility report, $24.6$ billion sensor nodes are expected to be installed by $2025$ \cite{Ericsson-report}. Majority of the applications in real time monitoring and industrial automation have strict delay requirements \cite{LowPowerTransceivers}. Therefore, increasing the lifetime of this massive number of sensor nodes while satisfying their latency requirements is a major challenge for 5G wireless networks. Low power transceivers with energy harvesting capability and intelligent medium access protocols have the potential to fulfill these diverse requirements. 

For perpetual energy, sensor nodes can harvest energy from the environment by using natural sources, i.e., wind, solar and vibration; magnetic coupling, i.e., inductive coupling and resonant magnetic coupling; and radio frequency (RF) sources \cite{harvest_10}. The dependence of natural sources on the environmental conditions results in an irregular and unpredictable amount of harvested energy, which is not acceptable for time critical applications. On the other hand, the usage of magnetic coupling is very restricted due to the large hardware size, short energy transfer range, and requirement for the alignment of energy transmitting and receiving coils. Due to full control on the energy transfer, small form factor, long range and no alignment requirements, the RF energy harvesting (EH) is a promising technique for the wireless energy transfer for IoT and MTC type applications. Based on the operational modes, RF-EH is classified into Simultaneous Wireless Power and Information Transfer (SWIPT) and Wireless Powered Communication Network (WPCN). 

SWIPT enables the sensor nodes to continuously decode the information and recharge their batteries by using the same signal\cite{SWIPT_mimo}. Generally, it is difficult to achieve the desired EH efficiency and information transfer capacity from the same signal. To achieve this trade-off, SWIPT is studied for different architectures, including power splitting, which divides the received signal into two streams for EH and information decoding (ID) with a pre-determined ratio; time switching, which uses the same antenna for EH and ID at different time intervals; and antenna separation, which uses different antennas for EH and ID \cite{SWIPT-Survey}. For a single user SWIPT, the trade-off between EH efficiency and information transfer capacity is studied for flat fading channel \cite{harvest_61}, additive white Gaussian noise channel \cite{harvest_06} and non-linear EH model \cite{harvest_new63}. The authors in \cite{SWIPT_multi-HAP-BatteryRechargeTime} study the battery recharge time of a single user with multiple energy transmitters. On the other hand, \cite{harvest_59} and \cite{harvest_19} consider multi-user SWIPT system to minimize transmit power and maximize energy transfer, respectively. For more detailed study on the SWIPT, we refer the reader to study recent survey papers presented in \cite{SWIPT-survey1,SWIPT-survey2,SWIPT-survey3}.

In WPCN, a hybrid access point (HAP) is responsible for energy transmission in the downlink and users harvest this energy to transmit back their information to the HAP by using the harvested energy. The authors in \cite{harvest_07} introduces the first protocol for the sum throughput maximization of a WPCN, in which the frame length is divided into two non-overlapping phases dedicated for energy and information transfer. This sum throughput objective results in an unfair resource allocation, motivating the researchers to incorporate different objective functions, such as weighted sum throughput maximization \cite{harvest_55}, total effective throughput maximization \cite{effective_throughput} and minimum throughput maximization \cite{harvest_04}. In this half-duplex model, as energy and information take place in different phases, the information transmission order of the users is not important, thus, no scheduling algorithm is needed. Since the energy harvesting rate of these networks is very low, the usage of multiple energy transmitters is investigated in \cite{WPCN-multiHAP1,WPCN-multiHAP2,WPCN-multiHAP3}. Particularly, \cite{WPCN-multiHAP1} studies the time allocation and load balancing for the throughput maximization, \cite{WPCN-multiHAP2} provides the implementation of a test-bed to maximize the amount of harvested energy and \cite{WPCN-multiHAP3} investigates a wireless body area network with the objective of maximizing the sum throughput. Following the usage of multiple energy transmitters to increase the energy harvesting rate, multi-cell WPCN, in which multiple HAPs are used for both energy transmission and information reception, has been proposed with the goal of increasing network coverage. Different objective functions have been studied, including sum throughput maximization and fair throughput optimization for non-orthogonal multiple access \cite{WPCN-multiHAP5_maxThroughputMin}, minimum throughput maximization \cite{WPCN-multiHAP4_minThroughputMax} and throughput maximization through beam-forming \cite{WPCN-multiHAP6_beamforming}. However, these objectives fail to provide any delay guarantee, which is necessary for the time critical networks. Only \cite{Elif_PIMRC} considers the minimum length scheduling problem. The usage of separate time intervals for energy and information transmission in half-duplex multi-cell WPCN leads to simpler and sub-optimal problem formulations. 

In the last few years, WPCN networks are extended for full-duplex by performing energy transmission and information reception at the same time and frequency to utilize the spectrum and time more efficiently. Thanks to the recent advances in self interference cancellation techniques \cite{harvest_41_ref26, harvest_41_ref27} and the related practical implementations \cite{harvest_30_ref19, harvest_30_ref26}, full-duplex has become a major transceiving technique for 5G and beyond networks \cite{harvest_new66}. In the context of single-HAP full duplex, the WPCN is mainly studied for throughput maximization by considering the whole network \cite{harvest_new68} or only the HAP \cite{harvest_30, harvest_40, harvest_50} operating in full duplex mode with perfect self interference cancellation. The residual self interference is considered in \cite{harvest_40, harvest_41, harvest_50,Shahid_ICC,DRSTMP_ITL}. Only in \cite{harvest_30,MLSP,onoff3,DRMLSP_AdhocNow}, the authors have considered the single HAP total transmission time minimization under a given traffic demand of the users for continuous rate without and with scheduling, constant and discrete transmission rate models, respectively. To the best of authors knowledge, the multi-HAP full-duplex WPCN is still not studied in the literature. 

The goal of this paper is to determine the optimal power control and scheduling for the constant and continuous transmission rate model with the objective of minimizing the schedule length for a realistic multi-cell full-duplex WPCN system. In full-duplex, all the users can harvest energy during the information transmission of other users. Therefore, the users with lower energy level can be scheduled later so that they harvest more energy and can transmit their information in shorter time. Moreover, in multi-cell WPCN, the users transmitting information simultaneously interfere with each other, which necessitates an intelligent grouping mechanism to minimize the interference among the concurrently transmitting users. The original contributions of this paper are listed below:

\begin{itemize}
\item We propose an optimization framework for the minimization of the schedule length in a full-duplex multi-cell WPCN in which multiple HAPs transmit energy to the users, and by using the harvested energy, multiple users transmit information to their respective HAP within a time slot, for the first time in the literature. The framework considers the concurrent transmissions of the users, and incorporates the non linear energy harvesting model for a full-duplex WPCN, in which the HAPs and users both operate in a full-duplex mode.  
\item We formulate two mixed integer non-linear optimization problems for the constant and continuous rate model to minimize the schedule length. The formulated optimization problems are non-convex and generally hard to solve for the global optimal solution. As a solution strategy, we perform decomposition of the optimization problems into power control problems and scheduling problems. First, we solve the power control problems for a given set of concurrently transmitting users optimally, and then, the optimal power control solutions are used to determine the optimal transmission time for a subset of users to be scheduled by the scheduling algorithms.
\item For the constant rate power control problem, we propose an optimal polynomial time algorithm based on the evaluation of the Perron-Frobenius conditions.
\item For the continuous rate power control problem, we propose an optimal polynomial-time algorithm based on the usage of bisection method and evaluation of the Perron-Frobenius conditions together.
\item For the constant rate scheduling problem, we propose a heuristic algorithm based on the maximization of the number of concurrently transmitting users within a transmission slot by maximizing the allowable interference on each user without violating their signal-to-noise-ratio (SNR) requirements. 
\item For continuous rate scheduling problem, we define a penalty function as a metric representing the advantage of the concurrent transmission of a set of users over their individual transmission. Through the analysis of the optimality conditions, we prove the correspondence between  the schedule length minimization and the minimization of the sum of penalties. Then, we propose a heuristic algorithm that allocates the users concurrently to minimize the sum penalties over the schedule. 
\item We evaluate the performance of the proposed algorithms in comparison to the successive transmission and concurrent transmission of randomly selected users for different HAP transmit powers, network densities and network size. 
\end{itemize}

The rest of the paper is organized as follows. Section \ref{sec:system} describes the multi-HAP in-band full-duplex WPCN system model and assumptions used throughout the paper. Section \ref{sec:mlsp} presents the mathematical formulation of the minimum length scheduling problem for constant and continuous transmission rate models. Section \ref{sec:PCP} provides the problem formulations and their solutions for the power control problem for constant and continuous transmission rate models. The scheduling algorithms for constant and continuous rate models are given in Section \ref{sec:scheduling}. Section \ref{sec:simulations} provides the performance evaluation of the proposed scheduling schemes in comparison to successive transmissions and concurrent transmission of a randomly selected set of users. Finally, the concluding remarks are given in Section \ref{sec:conclusion}.

\section{System Model and Assumptions} \label{sec:system}

The system model and related assumptions are described as follows:

 \begin{figure}[t]
 \centering
\includegraphics[width= 0.6 \linewidth]{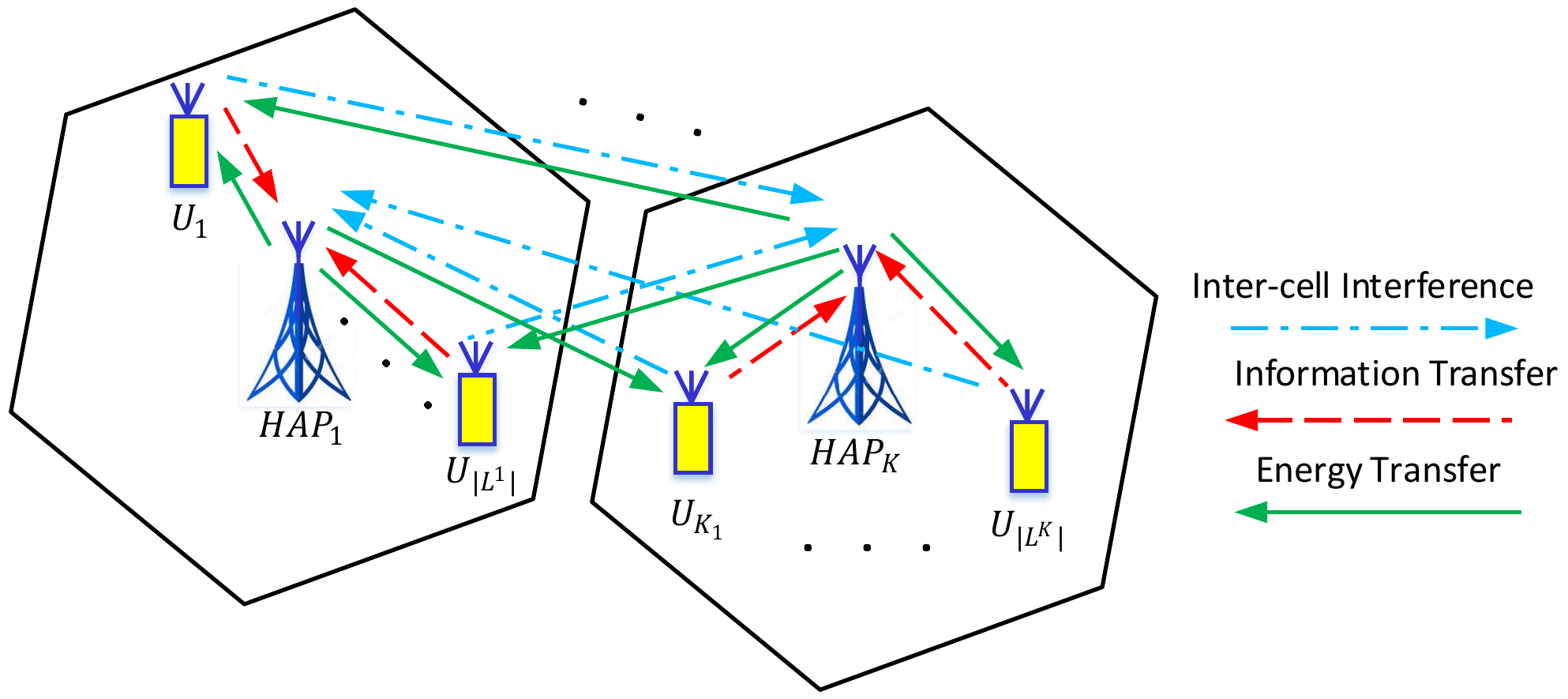}
\caption{Multi-cell wireless powered communication network architecture} \label{Figure:1}
\end{figure}
\begin{itemize}
\item We consider a multi-cell WPCN, which consists of $K$ HAPs and $N$ users, as depicted in Fig. \ref{Figure:1}. All the HAPs and users are equipped with a full-duplex antenna. Full-duplex antennas are used for concurrent energy transmission/reception and information reception/transmission. The HAPs are connected to a stable power line and transmit at a constant power $P_{h}$ continuously. On the other hand, users do not have any external power supply and can only harvest energy from the HAPs. The harvested energy is stored in a rechargeable battery with initial level $B_{n}$ for user $n$, where $n\in \{1,\cdots, N\}$, and capacity $B_{max}$. The users can harvest energy from all the HAPs but only transmit information to a single HAP to which they are connected to. To simplify the problem and focus on the power control and scheduling in the first step of the study, we consider a fixed topology, i.e., the users know the HAP they are connected to. The extension to the cell load balancing, in which the users are allocated evenly to all the cells, is subject to future work. The set of users connected to HAP $k$ is denoted by $L^k$ for $k\in\{1, \cdots, K\}$ such that $\bigcup_{k=1}^{K} L^{k}=N$. 
\item We assume that one HAP acts as central HAP. The central HAP is responsible to execute the algorithm for the joint power control, rate adaptation and scheduling based on the energy, data and channel characteristics of the users in each cell. Then, the central HAP communicates this information to the other HAPs, which supervise the scheduling and synchronization of all the users within their coverage area. 
\item We consider time division multiple access (TDMA) protocol within each cell for the uplink information transmission. The total time in which system remains operational is partitioned into frames and each frame is further divided into $M$ variable length non-overlapping time slots. In each time slot, a subset of users from different cells transmit their information simultaneously to their respective HAP. No intra-cell interference exists except the self interference at the HAP due to full duplex operational mode. However, simultaneous transmitting users from different cells may create interference to each other. 
\item The uplink information transmission and downlink energy transfer channels are assumed to be different. The uplink channel gain from user $n$ to HAP $k$ is denoted by $g_{nk}$ and the downlink channel gain from HAP $k$ to user $n$ is denoted by $h_{nk}$. We assume that all the channels are quasi-static, i.e., the channel gains remain the same in the current frame and can vary independently in the next frame. This block fading channel assumption is commonly used in the WPCN formulations \cite{WPCN-multiHAP4_minThroughputMax,harvest_19, harvest_04, harvest_39_ref17, harvest_39, harvest_51, harvest_56}. We further assume that the channel state information is perfectly known at the HAP \cite{harvest_07,harvest_04,harvest_50, harvest_44, harvest_55,harvest_51,harvest_30,harvest_40}. 
\item We assume that user $n$ has a traffic demand $D_n$ bits to be transmitted in each frame.
\item In the constant transmission rate model, we assume that each user allocated to the transmission slot $m$ transmits information at constant rate $r$ if the SNR of user $n$ is above a fixed threshold $\gamma_n$ as
\begin{equation}
\frac{P_{n}g_{nn}}{N_oW+W\sum_{j\neq n}P_{j}g_{jn}+\beta P_{h}} \geq \gamma_n, \hspace*{0.2cm} n \in \{1, \cdots, N\}
\end{equation}
where $P_{n}$ is the transmit power of user $n$, $W$ is the channel bandwidth, $N_o$ is the noise density, $\beta P_h$ represents the interference created by the energy radiation by the HAPs including self-interference and interference from the other HAPs and the term $W\sum_{j\neq n}P_{j}g_{jn}$ is the interference at the HAP to which user $n$ is connected from other users that are concurrently transmitting information in slot $m$.
\item In the continuous rate model \cite{MLSP, harvest_new62,harvest_07,harvest_30,harvest_50}, the Shannon channel capacity formulation for additive white Gaussian noise (AWGN) channel is used to calculate the maximum achievable data rate as a function of SNR by $x_{n}^{m}=Wlog_2(1+k_{n}^{m}P_{n})$, where $x_{n}^{m}$ denotes the transmission rate of user $n$ during the transmission time slot $m$, and $k_{n}^{m}$ is defined as $g_{nn}/(N_oW+W\sum_{j\neq n}P_{j}g_{jn}+\beta P_{h})$.
\item We assume a continuous power model in which user $n$ can transmit at any feasible power below a maximum power limit $P_{max}$, which is imposed to reduce the interference to other users.  
\item We consider a realistic non-linear energy harvesting model \cite{NLEH_model_01,NLEH_model_02,NLEH_01}, which uses a logistic function to incorporate the non-linear behavior of the energy harvesting circuitry and performs close to the experimental results presented in \cite{NLEH_Prac_01,NLEH_Prac_02,NLEH_Prac_03}. The total power received by user $n$ is $P_{n}^{T}=\sum_{k=1}^{K} h_{nk}P_h$. In this model, the energy harvesting rate for user $n$ is given by: 
\begin{equation}
C_n=\dfrac{P_s[\Psi_n-\Omega_n]}{1-\Omega_n},
\end{equation}
where $\Omega_n=(1+e^{a_nb_n})^{-1}$, is a constant to guarantee zero-input zero-output response; $P_s$ is the maximum harvested power during saturation; and $\Psi_n=(1+e^{-a_n(P_n^T-b_n)})^{-1}$ is the logistic function related to user $n$.
$a$ and $b$ are the positive constants related to the non-linear charging rate with respect to the input power and turn-on threshold, respectively. For a given energy harvesting circuit, the parameters $P_s$, $a$ and $b$ are determined by curve fitting.
\end{itemize}

\section{Minimum Length Scheduling Problem Formulation} \label{sec:mlsp}
In this section, we introduce the multi-cell minimum length scheduling problem for the constant and continuous transmission      rate models. We first present the mathematical formulations for both transmission rate models as mixed integer non-linear optimization problems. Then, we propose the solution strategy followed in the subsequent sections.

\subsection{Mathematical Formulation}
The joint optimization of the power control and scheduling for concurrently transmitting users with the objective of minimizing the schedule length with the traffic demand, energy causality and user transmit power constraints is formulated for constant and continuous transmission rate models.

\subsubsection{Constant Transmission Rate Model}
The formulation of the constant transmission rate model is as follows:

$\cal{MC-MLSP}$:
\begin{subequations} \label{opt_problem}
\allowdisplaybreaks
\small
\begin{align}
& \textit{minimize}
& & \sum_{m=1}^{M}\tau^m \label{obj}\\
& \textit{subject to}
& & \sum_{m=1}^{M}z_n^m=1, \hspace*{0.2cm} n \in \{1, \cdots, N\} \label{user_slot_assignment}\\
&&& \sum_{n\in L^k}z_n^m \leq 1, \hspace*{0.2cm} m \in \{1, \cdots, M\}, k \in \{1, \cdots, K\},  \label{user_HAP_assignment}\\
&&& x_n^m \tau^m \geq D_n \hspace*{0.2cm} n \in \{1, \cdots, N\}, \hspace*{0.2cm} m \in \{1, \cdots, M\}, \label{trafficDemand}\\
&&& x_n^m=z_{n}^{m}r \hspace*{0.2cm} n \in \{1, \cdots, N\}, \hspace*{0.2cm} m \in \{1, \cdots, M\}, \label{fixRate}\\
&&& \frac{P_{n}g_{nn}}{N_oW+W\sum_{j\neq n}P_{j}g_{jn}+\beta P_{h}} \geq \gamma_n, \hspace*{0.2cm} n \in \{1, \cdots, N\},\label{constantRate}\\
&&& B_{n}+\sum_{j=1}^{m} C_{n}\tau^{j}-P_{n}\tau^{m}\geq 0,  \hspace*{0.2cm} n \in \{1, \cdots, N\}, \label{energy_causality}\\
&&& P_{n}\leq z_n^m P_{max}, \hspace*{0.2cm} n \in \{1, \cdots, N\},\hspace*{0.2cm} m \in \{1, \cdots, M\}, \label{pmax} \\
& \textit{variables}
& & P_{n} \geq 0, \hspace*{0.1cm} \tau^m\geq 0, \hspace*{0.1cm} z_{n}^{m} \in \{0,1\}, \hspace*{0.1cm} n \in \{1, \cdots, N\}, \hspace*{0.1cm} m \in \{1, \cdots, M\}.\label{mlsp_vars}
\end{align}
\end{subequations}
The variables of the problem are $P_{n}$, the transmit power of user $n$; $\tau^m$, the length of transmission slot $m$; and $z_{n}^{m}$, a binary variable which takes a value $1$ if user $n$ is allocated to slot $m$ and $0$ otherwise. 

The objective of the optimization problem is to minimize the schedule length as given by Eq. (\ref{obj}). Eq. (\ref{user_slot_assignment}) states that each user should be allocated to only one time slot. Eq. (\ref{user_HAP_assignment}) represents that at most one user can be allocated to a slot for the same HAP. Eqns. (\ref{trafficDemand}) and (\ref{fixRate}) together represent the traffic demand of users. Eq. (\ref{constantRate}) represents the condition for SNR to achieve a constant rate. Eq. (\ref{energy_causality}) gives the energy causality constraint: The total amount of available energy, including both the initial energy and the energy harvested until and during the transmission of a user, should be greater than or equal to the energy consumed during its transmission. Eq. (\ref{pmax}) represents the maximum transmit power constraint for the users.

This optimization problem is a Mixed Integer Non-Linear Programming (MINLP) problem, thus, difficult to solve for the global optimum \cite{Boyd}.

\subsubsection{Continuous transmission rate model} 
The formulation of the continuous transmission rate model is very similar to the constant transmission rate model except Eqns. (\ref{trafficDemand}), (\ref{fixRate}) and (\ref{constantRate}) are replaced by the following constraint:
\begin{equation}
\sum_{m=1}^{M}z_{n}^{m}W \tau^m log_2\bigg(1+\frac{P_{n}g_{nn}}{N_oW+W\sum_{j\neq n}P_{j}g_{jn}+\beta P_{h}}\bigg)\geq D_n,  \hspace*{0.2cm} n \in \{1, \cdots, N\}
\end{equation}

The continuous rate problem is a MINLP problem, which is difficult to solve for the global optimal solution \cite{Boyd}.

\subsection{Solution Framework}
As both optimization problems are MINLP problems, finding the global optimal solution requires exponential time algorithms, which are intractable even for medium size networks. To overcome this intractability, we decompose the optimization problem into optimal power control problem and scheduling problem as detailed below:
\begin{itemize}

\item For a given set of concurrently transmitting users, we formulate the optimization problem to determine the minimum transmission slot length and the corresponding transmit power vector while considering the maximum transmit power constraint, traffic demand and energy causality of the users. We first show that the constant rate power control problem is a feasibility problem, and then, by using the Perron-Frobenius condition, we find the optimal power vector if the problem is feasible. For the continuous transmission rate problem, we first show that the transmission slot length should be equal for all the users, this reduces the problem to finding the power vector and the length of the transmission slot. We use the Perron-Frobenius conditions and bisection method together to find the power vector and transmission slot length.
\item Determining the power vector and transmission slot length for a given set of users reduces the problem to the optimization of the concurrently transmitting users within a transmission slot. For the constant rate scheduling problem, we exploit the affordable interference levels of all the users and group them based on these interference levels while considering the maximum transmit power and energy causality of the users. For the continuous rate scheduling problem, we introduce a penalty function for concurrently transmitting users as the difference between the concurrent transmission time and the sum of individual minimum transmission times. Then, based on the optimality conditions derived by using the penalty function and demonstration of the equivalence between total transmission time minimization and minimization of the sum of penalties, we propose a heuristic algorithm that intelligently allocates the users for concurrent transmission.  
\end{itemize}

\section{Power Control Problem}\label{sec:PCP}
In this section, we determine the optimal power control, rate allocation and time slot length for a given set of concurrently transmitting users from different cells with the objective of minimizing the time slot length for constant and continuous transmission rate models. For simplification, we remove the superscript $m$ in Section \ref{sec:mlsp}, which represents the time slot index.

\subsection{Constant Rate Transmission Model} \label{Sec:ConstantRateModel}
The optimal power control problem for constant transmission rate model is formulated as follows:
\begin{subequations} \label{pcp_problem}
\allowdisplaybreaks
\small
\begin{align}
& \textit{minimize}
& & t \label{pcp-obj1}\\
& \textit{subject to}
& & \frac{D_n}{x_n} \leq t, \hspace*{0.2cm} n \in \mathbf{S}, \label{traffic-Demand1}\\
&&& B_{n}+ C_{n}t-P_{n}t\geq 0,  \hspace*{0.2cm} n \in \mathbf{S}, \label{pcp-Ecausality1}\\
&&& P_{n}\leq P_{max}, \hspace*{0.2cm} n \in \mathbf{S}, \label{pcp-pmax1} \\
&&& x_n=r, \hspace*{0.2cm} n \in \mathbf{S}, \label{pcp-rate1} \\
&&& \frac{P_{n}g_{nn}}{N_oW+W\sum_{j\neq n}P_{j}g_{jn}+\beta P_{h}} \geq \gamma_n, \hspace*{0.2cm} n \in \mathbf{S} \label{pcp-SNR1}\\
& \textit{variables}
& & P_{n} \geq 0, \hspace*{0.1cm} t> 0, n \in \mathbf{S}.\label{pcp_vars1}
\end{align}
\end{subequations}
The objective of this problem is to determine the optimal time slot length for a given set of concurrently transmitting users in $\mathbf{S}$. The constraints in Eqns. (\ref{traffic-Demand1})-(\ref{pcp-SNR1}) represent the traffic demand, energy causality, maximum transmit power and constant transmission rate constraint, respectively. Note that $B_n$ refers to the energy available at the start of the transmission slot and we used the same notation for simplicity.   

This optimization problem can be reduced to a feasibility problem. If a feasible power vector exists, then, the optimal time slot length is $t=\max \limits_{n\in \mathbf{S}} D_n/r$. In the following, we present the feasibility of the concurrent transmission of the users in set $\mathbf{S}$. For the feasibility of concurrent transmission, there should exist a set of transmit power levels $P_n, \forall n\in \mathbf{S}$, such that the SNR constraint given in Eq. (\ref{pcp-SNR1}) is satisfied for each user in $\mathbf{S}$ without violating the maximum transmit power and energy causality constraints given in Eqns. (\ref{pcp-pmax1}) and (\ref{pcp-Ecausality1}). At any decision time $t_{dec}$, the existence of such a set of transmit power levels that satisfy the $P_{max}$ constraint can be determined by using Perron-Frobenius conditions described as follows: Let $\mathbf{G_\mathbf{S}}$ be a $\vert \mathbf{S} \vert \times \vert \mathbf{S} \vert$ relative channel gain matrix, in which $a_{ij}=g_{ji}/g_{ii}$, represents the entry of the $i^{th}$ row and $j^{th}$ column for $i \neq j$ and $a_{ij}=0$ for $i=j$. Let $\mathbf{D_{\mathbf{S}}}$ represent a $\vert \mathbf{S} \vert \times \vert \mathbf{S} \vert$ diagonal matrix with the $i^{th}$ entry equal to $\gamma_i$. Let $\mathbf{\sigma_{\mathbf{S}}}$ be a $\vert \mathbf{S} \vert \times 1$ normalized noise power vector with the $i^{th}$ entry equal to $\gamma_iWN_0/g_{ii}$. Then, Perron-Frobenius conditions states that there exists a power allocation vector that satisfies the $P_{max}$ constraint if and only if the largest real eigenvalue of $\mathbf{D_{\mathbf{S}}G_{\mathbf{S}}}$ is less than $1$ and every element of the component-wise minimum power vector $(\mathbf{I}-\mathbf{D_SG_S})^{-1}\mathbf{\sigma_S}$ is less than or equal to $P_{max}$, i.e., $(\mathbf{I}-\mathbf{D_SG_S})^{-1}\mathbf{\sigma_S}\preceq \mathbf{P_{max}}$, where $\mathbf{P_{max}}$ is a vector with all entries equal to $P_{max}$. For a feasible solution, the users must also satisfy the energy causality constraint, therefore, if the transmit powers also follow the energy causality constraint, the solution is feasible at the time of decision. 
\begin{lemma}
If $\mathbf{P^{min}}=(\mathbf{I}-\mathbf{D_SG_S})^{-1}\mathbf{\sigma_S}$ is an infeasible solution to Problem \ref{pcp_problem}, then, there is no feasible solution.

\end{lemma}

\begin{proof}
Let $\mathbf{P^{min}}=\{P_{1}^{min}, \cdots, P_{\vert \mathbf{S} \vert}^{min}\}$ denote the minimum power vector for a set $\mathbf{S}$. Any other power vector $\mathbf{P^{'}}$ satisfying the SNR constraint is component-wise greater than $\mathbf{P^{min}}$, i.e.,  $\mathbf{P^{'}}\succ \mathbf{P^{min}}$. Then, if $\mathbf{P^{min}}$ is infeasible, there exists either $P_{k}^{min}>P_{max}$ for any particular user $k \in \mathbf{S}$ or required energy $E_j$ is greater than the available energy, i.e., $E_{j}=P_{j}^{min}D_j/r>E_{j}^{available}$ for any particular use $j$. Then, if $\mathbf{P^{min}}$ is infeasible, there is no feasible solution for Problem \ref{pcp_problem}, since either $P_{k}^{'}$ will be greater than $P_{max}$ or $E_{j}^{'}$ will be still greater than $E_{j}^{available}$, making $\mathbf{P^{'}}$ also infeasible.
\end{proof}
The following theorem states that a feasible $\mathbf{P^{min}}$ vector is an optimal solution to the optimization problem  \ref{pcp_problem}.
\begin{theorem}
Let $\mathbf{P^{min}}=(\mathbf{I}-\mathbf{D_SG_S})^{-1}\mathbf{\sigma_S}$ denote the minimum feasible power vector obtained from Perron-Frobenius conditions. Then, if $\mathbf{P^{min}}$ is feasible, i.e., $\mathbf{P^{min}}\preceq \mathbf{P_{max}}$, and it satisfies the energy causality constraint of each user, then, $\mathbf{P^{min}}$ is an optimal solution to optimization problem \ref{pcp_problem}.
\end{theorem}

\begin{proof}
Any feasible power vector gives the same solution to problem \ref{pcp_problem}, i.e., $t=\max \limits_{n\in \mathbf{S}} D_n/r$. Thus, if $\mathbf{P^{min}}$ is a feasible, it is an optimal solution.
\end{proof}
Since the computational complexity of evaluating Perron-Frobenius conditions for $\vert \mathbf{S} \vert$ users is $\mathcal{O}(\vert \mathbf{S} \vert^3)$, the computational complexity of solving the power control problem for constant transmission rate model is $\mathcal{O}(\vert \mathbf{S} \vert^3)$. 
\subsection{Continuous Transmission Rate Model}
The power control problem for the continuous transmission rate model is very similar to the constant transmission rate model problem except constraints (\ref{traffic-Demand1}), (\ref{pcp-rate1}) and (\ref{pcp-SNR1}) are replaced by the constraint given by 

\begin{equation}
\frac{D_n}{Wlog_2\bigg(1+\frac{P_ng_{nn}}{N_oW+W\sum_{j\in \mathbf{S}, j\neq n}P_{j}g_{jn}+\beta P_{h}}\bigg)} \leq t, \hspace*{0.2cm} n \in \mathbf{S}, \label{traffic-Demand2}
\end{equation}
 
Since higher transmission rate results in smaller transmission time, the optimal rate allocation is the maximum feasible achievable rate to minimize the length of transmission time slot. Therefore, power control and rate allocation problem can be simply reduced to a power control problem in which the rates of the users are simply a function of the transmission power. 
In the following, we first illustrate that in the optimal solution, transmission time of all the users are equal, i.e.,  $t_n^*=t^*$ for $\forall n \in \mathbf{S}$, and then, we propose an algorithm that searches for the optimal transmission time slot length by using the bisection method. 

%

\begin{lemma}\label{Power_Preservation}
If $\mathbf{P^{min}}=(\mathbf{I}-\mathbf{D_SG_S})^{-1}\mathbf{\sigma_S}$ vector does not satisfy the energy causality constraint, any feasible power vector $\mathbf{P^{'}} \succeq \mathbf{P^{min}}$ does not satisfy the energy causality.
\end{lemma}

\begin{proof}
Since the Perron-Frobenius conditions gives the component-wise minimum power vector $\mathbf{P^{min}}=(\mathbf{I}-\mathbf{D_SG_S})^{-1}\mathbf{\sigma_S}$, if this power vector does not satisfy the energy causality for any user $n$, any other power vector $\mathbf{P^{'}}$ with the $n$-th element $P_{n}^{'}>P_{n}^{min}$ does not satisfy energy causality constraint. The energy consumed by user $n$ to transmit  $D_n$ bits by using transmit power $P_n$ is $E_n=P_nt_n$, where $t_n=D_n/Wlog_2(1+k_nP_n)$, resulting in $E_n=P_nD_n/Wlog_2(1+k_nP_n)$. It can be easily shown that $dE_n/dP_n > 0$ for $\forall P_n>0$, which guarantees that consumed energy is an increasing function of the transmit power. Hence, higher transmit power results in higher energy consumption for any user. Furthermore, the harvested energy during transmission time also reduces due to the shorter transmission time. This is still a violation of energy causality.
\end{proof}
\begin{theorem}
Let $t^*=\max\limits_{n \in \mathbf{S}}t_n^*$ denote the optimal transmission slot length, where $t_n^*$ is the optimal transmission time of user $n$. Then, the transmission time of all the users should be equal, i.e.,  $t_n^*=t^*$ for $\forall n \in \mathbf{S}$.
\end{theorem}

\begin{proof}
This can be proved by contradiction. Let $\{t_1^*, t_2^*, \cdots, t_{\vert \mathbf{S} \vert}^*\}$ denote the set of the optimal transmission times of the users in $\mathbf{S}$, leading to an optimal slot length $t^*=\max\limits_{j \in \mathbf{S}}t_j^*$, such that $t_n^*<t^*$ for a particular user $n\in \mathbf{S}$. Suppose that the transmit power of user $n$ is decreased by an arbitrarily small amount such that resulting transmission time $t_{n}^{'}$ is less than $t^*$ but greater than $t_n^*$; i.e., $t_n^*<t_{n}^{'}<t^*$. Then, the energy causality is not violated, since the consumed energy is an increasing function of the transmit power and the harvested energy increases due to longer transmission time as stated in Lemma \ref{Power_Preservation}. Obviously, the $P_{max}$ constraint is also not violated as we have reduced the transmit power. Moreover, since the transmission rate of any user $j$ is a decreasing function of the transmit power of user $n$, then for each user $j \neq n$ the resulting transmission time $t_j^{'}<t_j^*$ due to the reduced amount of interference caused by user $n$. The new time slot length $t^{'}=\max\limits_{j \in \mathbf{S}}t_j^{'}<\max\limits_{j \in \mathbf{S}}t_j^{*}=t^*$. This is a contradiction.      
\end{proof}

If we fix the length of the transmission slot, the optimal power control problem reduces to the feasibility problem of whether a set of power levels of the users in $S$ exists such that they satisfy the $P_{max}$ constraint without violating their energy causality constraints. The feasibility of such a set of transmit powers is determined by testing the Perron-Frobenius conditions as explained in Section \ref{Sec:ConstantRateModel}.

\begin{lemma}\label{lower-infeasibility}
For a concurrent transmission of a set of users $\mathbf{S}$, if a transmission time slot $t$ is infeasible, any value lower than $t$ is also infeasible.
\end{lemma}

\begin{proof}
For a set of users $\mathbf{S}$ transmitting concurrently in a transmission slot length $t$, the infeasibility may arise due to the violation of $P_{max}$ constraint or violation of energy causality constraint. If the infeasibility is due to $P_{max}$ constraint, any lower transmission slot length requires more power for information transmission, hence, still violates the constraint. If the infeasibility is due to energy causality constraint, lower transmission slot length results in lower harvested energy and higher transmit power, resulting in more energy consumption based on Lemma \ref{Power_Preservation}, therefore, still violates the constraint.
\end{proof}
 Lemma \ref{lower-infeasibility} allows the development of a fast bisection based algorithm (FBA) to determine the value of the optimal transmission slot length by halving the search space in each iteration. 
\begin{algorithm}[t]
	\caption{\small Fast bisection based algorithm (FBA)}\label{alg:fba}
	\begin{algorithmic}[1]	
	\small
	\Statex \textbf{Input}: A set of users $\mathbf{S}$
	\Statex \textbf{Output}: Transmission slot length $t^*$, transmit power vector $\mathbf{P^*}$
	\STATE $t_u \leftarrow \max \limits_{n\in \mathbf{S}} \big(D_n/Wlog_2(1+ P_{max}g_{nn}/(N_0W+W\sum_{j\neq n}P_{max}g_{jn}+\beta P_h))\big)$;
	\STATE $t_l \leftarrow \max \limits_{n\in \mathbf{S}} \big(D_n/Wlog_2(1+ P_{max}g_{nn}/(N_0W+\beta P_h))\big)$;
		\WHILE {$((t_u-t_l)/t_l)>\epsilon$}
		\STATE $t^*=(t_u+t_l)/2$;
		\STATE $\mathbf{P^{*}}=(\mathbf{I}-\mathbf{D_SG_S})^{-1}\mathbf{\sigma_S}$
		\IF {\text {the largest real eigenvalue of }$\mathbf{D_SG_S}<1$}	
		\IF {$\mathbf{P^{*}} \preceq  \mathbf{P_{max}}$ \text{and} $t^*$ \text{and} $\mathbf{P^{*}}$ \text{satisfy Eq. (\ref{pcp-Ecausality1})}}
		\STATE $t_u=t^*$;
		\ELSE
		\STATE  $t_l=t^*$;
		\ENDIF
		\ENDIF
		\ENDWHILE
		\STATE $t^*=t_u$
	\end{algorithmic}
\end{algorithm}  
The input of FBA algorithm is a set of users $\mathbf{S}$. The outputs include the optimal transmission slot length $t^*$ and the optimal transmit power vector $\mathbf{P^{*}}$. FBA determines an interval for the optimal value of the transmission slot length corresponding to the transmission of $D_i$ bits for each user $i \in \mathbf{S}$, and then, reduces the length of this interval by half in each iteration. The upper limit $t_u$ and lower limit $t_l$ of the transmission slot length $t$ are determined by maximum interference, i.e., all users in set $\mathbf{S}$ are transmitting at power $P_{max}$, and zero interference, i.e., all users have zero transmit power except user $n$, respectively (Lines 1-2). In each iteration, the algorithm checks the feasibility of the middle point of the interval ($t_u$, $t_l$) (Lines 5-7). For the feasibility, it evaluates the power vector corresponding to transmission slot length $t^*$ (Line 5); checks the Perron-Frobenius conditions one by one, i.e., the largest eigenvalue of $\mathbf{D_SG_S} < 1$ and every element of the power vector $\mathbf{P^{*}}$ is less than or equal $P_{max}$, i.e., $\mathbf{P^{*}} \preceq  \mathbf{P_{max}}$; and then, checks the energy causality constraint for each user (Lines 6-7). If the middle point $t^*$ is feasible, then any point greater than $t^*$ can be discarded from the search space by updating the upper limit (Line 8). On the other hand, if the middle point $t^*$ is infeasible, any value lower than this middle point is also infeasible, as shown in Lemma \ref{lower-infeasibility}, necessitating the search in interval ($(t_l+t_u)/2$, $t_u$) by updating the lower limit (Line 10). The algorithm terminates when it achieves a pre-defined relative error bound denoted by $\epsilon$ (Line 3). The value of the transmission slot length is then set to $t_u$ (Line 14). Since the optimal transmission slot length exists in the interval ($t_l$,$t_u$), $t_l$ may give an infeasible solution point. 

The computational complexity of the standard bisection method is $K_B=\lceil log((t_{u}^{int}-t_{l}^{int})/\epsilon t_{l}^{int})\rceil$, where $t_{l}^{int}$ and $t_{u}^{int}$ are the initial lower and upper limits set by the algorithm in (Lines 1-2), respectively. Therefore, the computational complexity of FBA is $O(K_B \vert \mathbf{S} \vert^3)$ due to the evaluation of the Perron-Frobenius conditions in each iteration.
\section{Scheduling}\label{sec:scheduling}
In Section \ref{sec:PCP}, we solved the power control problem optimally for a known set of users transmitting information simultaneously. For all the users transmitting concurrently, the information transmission may be too long due to extensive interference caused to each other. Therefore, the schedule length can further decrease by intelligently grouping the users into a set of subsets for concurrent transmission. However, determining the optimal subsets of the users transmitting concurrently and their schedule require an exponential computational effort. Let $\nu=\{\nu_j:1 \leq j \leq \vert \nu \vert\}$ denote the set of all possible partitions of the set of users $\mathcal{N}=\{ 1, \cdots, N\}$, where a partition of set $\cal{N}$ is defined as a set of nonempty mutually exhaustive and mutually exclusive subsets of set $\cal{N}$, i.e., for the set of $10$ users, $\{\{1, 5, 8\},\{4, 7, 9, 10\},\{2, 3, 6\}\}$ is one possible partition of set $\cal{N}$. The total number of such partitions for a set of $N$ users is $\vert \nu \vert=\sum_{j=0}^{N-1}{N-1 \choose j}B_{j}$, where $B_j$ is a Bell number with $B_0=1$ and $B_j=\sum_{z=0}^{j-1}{j-1 \choose z}B_{z}$. For a given partition of $\cal{N}$, each subset of users transmit information concurrently  within a time slot. Therefore, for a specific partition $\nu_j$ with $\zeta_j$ number of subsets, the total number of possible transmission sequences are $\zeta_j!$, i.e., for the above mentioned example $\zeta_j=3$ resulting in 6 possible transmission sequences. Then, total number of possible transmission sequences for a set of users $\cal{N}$ are $\sum_{j=1}^{\vert \nu \vert} \zeta_j !$. This exponential computational complexity of the scheduling problem makes the problem intractable even for a medium size networks. Hence, scalable and fast solutions are required. In the following, we present the scheduling problems for the constant and continuous rate transmission models.

\subsection{Scheduling for Constant Rate Transmission Model}

In this section, we present a heuristic scheduling algorithm for the constant transmission rate model, which aims at maximizing the number of users within each concurrently transmitting set by maximizing the interference on each user without violating the energy causality, SNR requirements and maximum transmit power constraint.

Let $I_{n}^{max}$ denote the maximum allowed interference level from other users that a particular user $n$ can afford, i.e., any interference level greater than $I_{n}^{max}$ results in an infeasible transmission scenario for user $n$. This maximum allowed interference can be determined by using Eq. (\ref{constantRate}) as follows:
\begin{equation}
I_{n}^{max} = \frac{1}{W}\bigg[\frac{P_ng_{nn}}{\gamma_n}-N_0W-\beta P_h\bigg],\label{max-allowed-interference}
\end{equation}
It is important to note that the maximum allowed interference level of any user is an increasing function of the transmit power of that user. Hence, increasing the transmit power allows the user to accommodate more interference and possibly allocation of more simultaneously transmitting users in a slot. Furthermore, for a user $n$, the earliest scheduling time $ts_{n}$ is defined as the first time instant when it can afford transmission at constant rate individually and is given by $ts_{n}=\max(0,P_{n}^{min}D_n/C_nr-B_n/C_n-D_n/r)$, 
where $P_{n}^{min}=min(P_{max},\gamma_n (N_0W+\beta P_h)/g_{nn})$, is the minimum power required to achieve the SNR $\gamma_n$ without any interference, which is derived from Eq. (\ref{pcp-SNR1}) for the case when the user transmits individually. The addition of any user in set $\mathbf{S}$ results in a delay in the earliest time instant $ts_{n}$ due to increased interference. 

The scheduling algorithm should determine the subsets of users such that the maximum number of users transmit their information simultaneously with the interference level close to the $I_{n}^{max}$ for each user $n$. Therefore, at any decision time $t_{dec}$, first we determine a set of users that can afford the constant rate at $t_{dec}$, i.e., the users for which $ts_{n} \leq t_{dec}$. Then, for this set of users, we evaluate the maximum power they can afford, i.e., $P_n = max(P_{max}, E_n(t_{dec})/(D_n/r))$ and the maximum level of interference they can afford from other users in a concurrent transmission, i.e., $I_{n}^{max}=\frac{1}{W}(P_n g_{nn}/\gamma_n-N_0W-\beta P_h)$. Then, CRSA evaluates the pairwise interference of the feasible users and group the users from different cells such that number of concurrently transmitting users within each time slot is maximized without violating the maximum allowed interference level. 
\begin{algorithm}[ht!]
	\caption{\small Constant Rate Scheduling Algorithm (CRSA)}\label{alg:crsa}
	\begin{algorithmic}[1]	
	\small
	\Statex \textbf{Input}: A set of users $\mathcal{N}$
	\Statex \textbf{Output}: Simultaneously transmitting user set in each slot $\mathbf{S}$, transmission time for each slot $\mathbf{t^*}$, transmit power vector $\mathbf{P^*}$
	\STATE  $m \leftarrow 1$, $t_{dec} \leftarrow 0$, $\mathbf{P} \leftarrow \mathbf{0}$, $\mathbf{I^{max}} \leftarrow \mathbf{0}$, $\mathbf{I} \leftarrow \mathbf{0}$ 	
	\STATE determine $ts_{n}$, for all $n \in \mathcal{N}$ by using Eq. (\ref{earliest-time}) 
\WHILE {$\mathcal{N} \neq \emptyset$}
	\STATE $\mathcal{F} \leftarrow \emptyset, \mathbf{S^m} \leftarrow \emptyset$
	\IF {$t_{dec}< \min \limits_{n\in \mathcal{N}}ts_{n}$}
		\STATE $t_{dec}= \min \limits_{n\in \mathcal{N}}ts_{n}$
	\ENDIF
	\FOR {$n \in \mathcal{N}$}
			\IF{$ts_{n} \leq t_{dec}$}
				\STATE $\mathcal{F} \leftarrow \mathcal{F} + n$
				\STATE $E_n(t_{dec})=B_n+C_n(t_{dec}+D_n/r)$
				\STATE $P_n \leftarrow min(P_{max}, E_n(t_{dec})/(D_n/r))$
				\STATE $I_{n}^{max} \leftarrow \frac{1}{W}\bigg[\frac{P_n g_{nn}}{\gamma_n}-N_0W-\beta P_h\bigg]$
			\ENDIF
	\ENDFOR
		\STATE evaluate interference  $I_{ns}$ of user $n$ to $s$, for all $n,s \in  \mathcal{F}$
		\STATE $u \leftarrow arg \max \limits_{j\in \mathcal{F}} I_{j}^{max} $
		\STATE $I_s \leftarrow 0$, for all $s \in  \mathcal{F}$
		\STATE $\mathbf{S^m} \leftarrow \mathbf{S^m} + \{u\}$
	\FOR {$k=1:K$}
		\IF {$u \notin$ $L_k$}
		\STATE $\mathcal{F}^{k} \leftarrow  L_k \cap \mathcal{F}$
		\STATE sort $\mathcal{F}^{k}$ in descending order of $I_{n}^{max}$
	\FOR {$v \in \mathcal{F}^{k}$}
		\IF {$ \sum_{j \in \mathbf{S^m}} I_{jv} \leq I_{v}^{max} $}
		\IF {$I_{s}^{max}-I_s \geq I_{vs}, \hspace*{0.2cm} \forall s \in \mathbf{S^m}$} 
			\STATE $I_v \leftarrow \sum_{j \in \mathbf{S^m}} I_{jv}$	
			\STATE $I_s=I_s+I_{vs}, \hspace*{0.2cm} \forall s \in \mathbf{S^m}$		
			\STATE $\mathbf{S^m} \leftarrow \mathbf{S^m}+\mathcal{F}^{k}(v)$
			\STATE \text{break,}
		\ENDIF
		
		\ELSE
			\STATE \text{break,}
		\ENDIF 
		\ENDFOR
		\ENDIF
	\ENDFOR
		\STATE $t_{m}^{*} \leftarrow \max \limits_{s \in \mathbf{S^m}} D_s/r, $
		\STATE $\mathcal{N} \leftarrow \mathcal{N}-\mathbf{S^m}$
		\STATE $t_{dec}=t_{dec}+t_{m}^{*}$
		\STATE $m \leftarrow m+1$
\ENDWHILE
	\end{algorithmic}
\end{algorithm} 

Constant Rate Scheduling Algorithm (CRSA) is given in Algorithm \ref{alg:crsa} and described in detail next. Let $\mathbf{S}=\{\mathbf{S^1}, \cdots, \mathbf{S^M}\}$ be a set of subsets of users concurrently transmitting in time slots $\mathbf{t^*}=\{t_{1}^{*}, \cdots, t_{M}^{*}\}$ by using powers $\mathbf{P^*}=\{\mathbf{P_{1}^{*}}, \cdots, \mathbf{P_{M}}\}$.  The $m^{th}$ entry of $\mathbf{S}$ denoted by $\mathbf{S^m}$ contains a subset of users that transmit simultaneously in the $m^{th}$ time slot with duration $t_{m}^{*}=\max \limits_{s \in \mathbf{S^m}} D_s/r$ and corresponding power vector $\mathbf{P_{m}^{*}}$. Let $\mathbf{I^{max}}$ be a $N \times 1$ maximum affordable interference vector, in which $j$-th entry represents the maximum allowed interference level from other users on user $j$. Let $\mathbf{I}$ be a $N \times K$ pairwise interference matrix, in which $I_{ij}$ is the entry of $i^{th}$ row and $j^{th}$ column representing the interference caused by user $i$ on the HAP $j$.
The input of CRSA is a set of users denoted by $\mathcal{N}$, whereas its outputs include a set $\mathbf{S}$, and corresponding time and power allocation vectors denoted by $\mathbf{t^*}$ and $\mathbf{P^*}$, respectively. The algorithm starts by initializing the transmission slot number $m$ to $1$, decision time $t_{dec}$ to $0$, power vector $\mathbf{P}$, maximum affordable interference $\mathbf{I^{max}}$ and interference matrix $\mathbf{I}$ to $\textbf{0}$ (Line 1). Then, it determines the earliest starting time $ts_n$ for every user $n, n \in \mathcal{N}$ (Line 2). In each iteration, the algorithm determines the set $\mathbf{S^m}$ as well as their powers. The $CRSA$ first updates the decision time by setting it equal to the minimum of earliest starting times of the users, since no user has enough energy for transmission earlier than that (Lines 5-7). Then, it checks the feasibility of each unallocated user based on their earliest scheduling time (Line 9). For the set of feasible users at time $t_{dec}$, denoted by $\mathcal{F}$, $CRSA$ evaluates the amount of energy each user needs to complete its transmission, maximum transmit power and the maximum interference each user can afford (Lines 11-13). Once the set of feasible users $\mathcal{F}$ and their parameters are determined, the algorithm groups the users in set $\mathcal{F}$ with similar interference characteristics for concurrent transmission. For this purpose, the algorithm evaluates the pairwise interference that each user in $\mathcal{F}$ causes to others while transmitting simultaneously (Line 16). $CRSA$ picks the user that can afford the maximum interference from other users, since this allows the addition of more users in the set of simultaneously transmitting users (Line 17). Then, it initializes the interference that is created on a certain user by concurrent transmissions $I_s$ to $0$ and adds the user that can afford maximum interference to the set $\mathbf{S^m}$ for transmission in the current time slot (Lines 18-19). The algorithm tries to group the most suitable user from each cell iteratively (Lines 20-37). For this, $CRSA$ visits each cell one by one and terminates the search when either a feasible user is found or no user can be grouped from the current cell. For each cell, $CRSA$ first sorts all the feasible users within the cell based on their $I_{n}^{max}$ values (Line 23). Then, it checks the compatibility for each user for possible concurrent transmission until it finds a suitable user from current cell. For a user $v$ to be compatible for concurrent transmission in the set $S^{m}$, the total interference caused by all the users in set $\mathbf{S^m}$ must be less than $I_{v}^{max}$, i.e., $ \sum_{j \in \mathbf{S^m}} I_{jv} \leq I_{v}^{max} $ (Line 25) and the interference caused by user $v$ must be less than the affordable interference margin of the existing users in set $\mathbf{S^m}$, i.e., $I_{s}^{max}-I_s \geq I_{vs}, \hspace*{0.2cm} \forall s \in \mathbf{S^m}$ (Line 26). If the first condition is not valid for the user with maximum $I_{n}^{max}$ value in current cell, this implies that no other user from current cell can afford this much interference, hence the algorithm moves to next cell without picking any user from this cell (Line 33). If the first condition is true but the user is causing excessive interference to the users in set $S^{m}$, the algorithm keeps on searching among the remaining feasible users within current cell for simultaneous transmission (Lines 24-32). If both of these conditions are true, the algorithm updates the interference levels of set $S^{m}$, adds this user to set $S^{m}$ and moves to the next cell (Lines 27-29). Once all the cells are visited, $CRSA$ evaluates the transmission time slot length for $S^{m}$, updates the set of unallocated users, decision time and the slot number (Lines 38-41). The algorithm terminates when all the users are allocated (Line 3).

\subsection{Scheduling for Continuous Rate Transmission Model}
In this section, we present a heuristic scheduling algorithm for the continuous transmission rate model. First, we define a penalty function for a set of users transmitting concurrently as the difference between their concurrent transmission time and the sum of individual minimum transmission times of the users. Then, we demonstrate the equivalence between the schedule length minimization objective and the sum of penalties minimization. We present a heuristic scheduling algorithm for continuous rate model based on the derivation of a penalty criterion for the addition of a new user to an existing set of users transmitting simultaneously. 

Let the minimum transmission time $t_{i}^{min}$, $i \in \{1,2, \cdots, \vert \mathcal{N} \vert \}$ be the transmission time for a user $i$ while transmitting individually at power $P_i=P_{max}$. Let $s$ be the starting time of the transmission for a concurrently transmitting set of users $\mathbf{S}$. Let $\tau_{\mathbf{S}}(s)$ be the time slot length for a set of concurrently transmitting users $\mathbf{S}$ evaluated at any decision time $s$ by using Algorithm \ref{alg:fba} considering the amount of energy available at decision time $s$. In the following, we define a new metric incorporating $\tau_{\mathbf{S}}(s)$ and $t_{i}^{min}$ as follows: 

\begin{definition} \label{def_penalty}
The penalty function of a set of users $\mathbf{S}$, $\rho_{\mathbf{S}}(s)$, is defined as the difference between the concurrent transmission time of set $\mathbf{S}$, denoted by $\tau_{\mathbf{S}}(s)$, and the sum of individual minimum possible transmission times of the users in set $\mathbf{S}$, and formulated as $\rho_{\mathbf{S}}(s)=\tau_{\mathbf{S}}(s)-\sum_{i=1}^{\vert \mathbf{S} \vert}t_{i}^{min}$.    
\end{definition}



In the foregoing definition, we extend the idea of the penalty function defined for a single cell scenario in \cite{MLSP} to a multi-cell scenario for concurrent transmission. 
For the scenario in which $\rho_\mathbf{S}(s)=0$, i.e., $\tau_{\mathbf{S}}(s)=\sum_{i=1}^{\vert \mathbf{S} \vert}t_{i}^{min}$, the concurrent transmission slot length is exactly equal to the TDMA schedule length in which all users transmit at $P_{max}$ individually, i.e., $P_i=P_{max}, \forall i\in \mathbf{S}$, which is the lower bound for the schedule length of a TDMA transmission model. If the penalty is negative, i.e., $\rho_{S}(s)<0$, it means that the concurrent transmission of set $\mathbf{S}$ is even better than the best TDMA schedule i.e., $\tau_{\mathbf{S}}(s) < \sum_{i=1}^{\vert \mathbf{S} \vert}t_{i}^{min}$, which motivates that the set of users $\mathbf{S}$ should be evaluated for concurrent transmission. The following Lemma states that for a given set of users $\mathcal{N}$ and a predetermined set of subsets $\mathbf{S^m}$ forming a partition of set $\mathcal{N}$, minimizing the sum of penalties is equivalent to minimizing the schedule length.  

\begin{lemma} \label{lemma_Obj_equal}
For a given set of users $\mathcal{N}$ and a partition of set $\mathcal{N}$ consisting of $M$ subsets $\mathbf{S^m}$, $m \in \{1,2, \cdots, M\}$, transmitting concurrently, the objective of minimizing the schedule length $ min\sum_{m=1}^{M}\tau_{\mathbf{S^m}}$ is equivalent to minimizing the sum of penalties $ min\sum_{m=1}^{M} \rho_{\mathbf{S^m}}(s) $.
\end{lemma}

\begin{proof}
By definition of the penalty function, 
\begin{equation} \label{penalty_1}
min\sum_{m=1}^{M} \rho_{\mathbf{S^m}}(s) = min \sum_{m=1}^{M} \tau_{\mathbf{S^m}}(s)-\sum_{m=1}^{M}\sum_{i\in \mathbf{S^m}}t_{i}^{min}.
\end{equation}
Since $\mathbf{S^m}, m \in \{1,2, \cdots, M\}$, is a partition of $\mathcal{N}$, this can be rewritten as follows:
\begin{equation} \label{penalty_2}
min\sum_{m=1}^{M} \rho_{\mathbf{S^m}}(s) = min \sum_{m=1}^{M} \tau_{\mathbf{S^m}}(s)-\sum_{i\in \mathcal{N}}t_{i}^{min}
\end{equation}
Since $t_i^{min}$ is constant for any user $i$, the term $\sum_{i\in \mathcal{N}}t_{i}^{min}$ can be removed from the objective function. Then, minimization of the schedule length is equivalent to minimization of the sum of penalties.
\end{proof}

In the following, we derive a criterion for the addition of any new user to a set of existing concurrently transmitting users. 
For a given set of users $\mathbf{S}=\{1,2, \cdots, \vert \mathbf{S} \vert\}$ transmitting concurrently, addition of any new user to set $\mathbf{S}$ causes additional interference for the existing users, hence, increases the concurrent transmission time. Let $\rho_{\mathbf{S}}(s)=\tau_{\mathbf{S}}(s)-\sum_{i=1}^{\vert \mathbf{S} \vert}t_{i}^{min}$ is the penalty of set $\mathbf{S}$ such that $\rho_{\mathbf{S}}(s)<0$, and suppose that an additional user $j$ is added to the set $\mathbf{S}$ such that resulting set is denoted by $\mathbf{S^{'}}= \mathbf{S}+ \{j\}$ with penalty $\rho_{\mathbf{S^{'}}}(s)=\tau_{\mathbf{S^{'}}}(s)-\sum_{i=1}^{\vert \mathbf{S^{'}} \vert}t_{i}^{min}$. If $\rho_{\mathbf{S^{'}}}(s)>\rho_{\mathbf{S}}(s)$, i.e., $\tau_{\mathbf{S^{'}}}(s)-t_{j}^{min}>\tau_{\mathbf{S}}(s)$, the difference between $\tau_{\mathbf{S^{'}}}(s)$ and $\tau_{\mathbf{S}}(s)$ is greater than $t_{j}^{min}$, which means that the addition of user $j$ results in an increase in the concurrent transmission time $\tau_{\mathbf{S}}(s)$ by an amount greater than $t_{j}^{min}$ due to its excessive interference on the existing users. Therefore, 
if the addition of a new user to the concurrently transmitting set of users increases the overall penalty of the set, the new user should not be allocated in the set $\mathbf{S}$ for the simultaneous transmission. On the other hand, if $\rho_{\mathbf{S^{'}}}(s) \leq \rho_{\mathbf{S}}(s)$, i.e., $\tau_{\mathbf{S^{'}}}(s)-t_{j}^{min} \leq \tau_{\mathbf{S}}(s)$, the addition of new user $j$ to set $\mathbf{S}$ is advantageous for our objective, since this addition reduces the concurrent transmission time of the users by an amount greater than $t_{j}^{min}$, which is the minimum transmission time for user $j$ individually, meaning that if the set $\mathbf{S}$ and user $j$ transmit separately, the resulting transmission time is definitely larger. 

\begin{lemma} \label{lemma_user_addition}
For a given set of users $\mathbf{S}$ transmitting simultaneously with penalty $\rho_{\mathbf{S}}(s)<0$ and a user $j \notin \mathbf{S}$ with penalty $\rho_j(s)=0$, the concurrent transmission of the set $\mathbf{S^{'}}=\mathbf{S}+\{j\}$ is minimized if and only if $\rho_{\mathbf{S^{'}}}(s)<\rho_{\mathbf{S}}(s)$.
\end{lemma}

\begin{proof}
According to the definition of penalty, $\rho_{\mathbf{S}}(s)=\tau_{\mathbf{S}}(s)-\sum_{i \in \mathbf{S}}t_{i}^{min}$ and $\rho_j(s)=0$ indicate that $\tau_j(s)=t_{j}^{min}$. Then, after the addition of user $j$ to the set $\mathbf{S}$ for concurrent transmission, the resulting penalty $\rho_{\mathbf{S^{'}}}(s)=\tau_{\mathbf{S^{'}}}(s)-\sum_{i \in \mathbf{S}}t_{i}^{min}-t_{j}^{min}$ which is equivalent to $\rho_{\mathbf{S^{'}}}(s)=\tau_{\mathbf{S^{'}}}(s)-\sum_{i \in \mathbf{S}}t_{i}^{min}-\tau_{j}(s)$. Then, the term $\rho_{\mathbf{S^{'}}}(s)<\rho_{\mathbf{S}}(s)$ is equivalent to $\tau_{\mathbf{S^{'}}}(s)<\tau_{\mathbf{S}}(s)+\tau_j(s)$, which means that the completion time of set $\mathbf{S}$ and user $j$ is minimized if and only if $\rho_{\mathbf{S^{'}}}(s)<\rho_{\mathbf{S}}(s)$.
\end{proof}

Lemma \ref{lemma_user_addition} suggests that at a particular decision time, if there is a zero penalty user that decreases the overall penalty, the optimal decision is to add this user in the concurrently transmitting set of users.

In the following, we present a penalty based scheduling algorithm (PSA) for the continuous rate transmission model based on the foregoing discussion, which is given in Algorithm \ref{alg:psa}. Let $L_{k^i}$ denote the set of users that are connected to the  same HAP as user $i$. The input of PSA is a set of users $\mathcal{N}$. The outputs are subsets of users transmitting simultaneously, transmit power of the users and the transmission time of each subset of users transmitting simultaneously. The algorithm starts by initializing the decision time $t_{dec}$ to $0$ and the transmission slot number $m$ to $1$ (Line $1$). In each iteration, PSA determines a subset of users $\mathbf{S^m}$, its corresponding time slot length and the powers of the users that transmit concurrently. The algorithm first determines the individual penalties of all the unallocated users at time $t_{dec}$ and sorts them in the increasing order of their penalties (Lines 3-4). The PSA stores these sorted users in a temporary array $\mathcal{N}_{sort}$. Then, the algorithm picks the user with minimum individual penalty as the first user for the simultaneous transmission and adds it to the set $\mathbf{S^m}$ followed by removing all the users from this cell, since only one user can transmit from a cell at one time (Lines $6-7$). Then, PSA checks whether the next lowest penalty user can be allocated with set $\mathbf{S^m}$ or not by comparing the penalty $\rho_{\mathbf{S^{'}}}(t_{dec})$ with $\rho_{\mathbf{S}}(t_{dec})$ as suggested by Lemma \ref{lemma_user_addition} and $\rho_{\mathbf{S^{'}}}(t_{dec})<0$, which means that concurrent transmission of set $\mathbf{S^{'}}$ is better than the separate transmission of set $\mathbf{S^m}$ and user $j$ (Lines $12$). The set $\mathbf{S^{'}}$ is the union of set $\mathbf{S}$ and the minimum penalty user of $\mathcal{N}_{sort}$. If both conditions hold, i.e., $\rho_{\mathbf{S^{'}}}(t_{dec})<\rho_{\mathbf{S^m}}(t_{dec})$ and $\rho_{\mathbf{S^{'}}}(t_{dec})<0$, the PSA adds this user to set $\mathbf{S^m}$ for simultaneous transmission and discards all the users from this cell (Line $13-14$), since only one user from a cell can transmit at a time. Otherwise, this user is not compatible with set $\mathbf{S^{m}}$ and as suggested by Lemma \ref{lemma_user_addition}, PSA discards it (Lines $16$). The PSA keeps on searching users for the simultaneous transmission until all the cells are visited (Lines $8-18$). Then, PSA evaluates the simultaneous transmission time and powers for set $\mathbf{S^m}$ and updates the set of unallocated users $\mathcal{N}$, decision time $t_{dec}$ and transmission slot number accordingly (Lines $19-22$). The algorithm terminates when all the users are allocated (Line $2$).  

The computational complexity of PSA is $\mathcal{O}(C_1N^2)$ for $N$ users, where $C_1$, the complexity of the power control algorithm given in Algorithm \ref{alg:fba}. The algorithm runs the power control algorithm to evaluate the penalties of the users in each iteration. Since the inner while loop discards one user in each iteration, there are $N$ executions at maximum, and the main while loop picks single user in each iteration for slot allocation, the total number of iterations is $N$. 

\begin{algorithm}[t]
	\caption{\small Penalty based Scheduling Algorithm (PSA)}\label{alg:psa}
	\begin{algorithmic}[1]	
	\small
	\Statex \textbf{Input}: A set of users $\mathcal{N}$
	\Statex \textbf{Output}: Simultaneously transmitting user in each slot $\mathbf{S^{m}}$, transmit power $\mathbf{P^m}$, transmission time for each slot $\tau_{\mathbf{S^m}}$, for $m \in \{1, \cdots, M\}$, 
\STATE  $t_{dec} \leftarrow 0$, $m \leftarrow 1$, 
\WHILE {$\mathcal{N} \neq \emptyset$}
\STATE determine $\rho_i(t_{dec}), i\in \mathcal{N}$
\STATE $\mathcal{N}_{sort} \leftarrow$ sorted users of $\mathcal{N}$ in increasing order of their penalties $\rho_i(t_{dec}), \forall i \in \mathcal{N}$,
\STATE $\mathbf{S^m} \leftarrow \emptyset$, 
\STATE $\mathbf{S^{m}} \leftarrow \mathbf{S^{m}}+\{\mathcal{N}_{sort}(1)\}$
\STATE $\mathcal{N}_{sort} \leftarrow \mathcal{N}_{sort}-L_{k^1}$
\WHILE {$\mathcal{N}_{sort} \neq \emptyset$}
\STATE $\mathbf{S^{'}} \leftarrow \emptyset$
\STATE $i \leftarrow \mathcal{N}_{sort}(1)$
\STATE $\mathbf{S^{'}} \leftarrow \mathbf{S^m}+\{i\}$
\IF {$\rho_{\mathbf{S^{'}}}(t_{dec})<\rho_{\mathbf{S^m}}(t_{dec})$ and $\rho_{\mathbf{S^{'}}}(t_{dec})<0$}
\STATE $\mathbf{S^{m}} \leftarrow \mathbf{S^{m}}+\{i\}$
\STATE $\mathcal{N}_{sort} \leftarrow \mathcal{N}_{sort}-L_{k^i}$
\ELSE 
\STATE $\mathcal{N}_{sort} \leftarrow \mathcal{N}_{sort}-\{i\}$
\ENDIF
\ENDWHILE
\STATE evaluate $\tau_{\mathbf{S^m}}(t_{dec})$ and $\mathbf{P^m}$ for set $\mathbf{S^m}$ by using Algorithm \ref{alg:fba} 
\STATE $\mathcal{N} \leftarrow \mathcal{N}-\mathbf{S^{m}}$
\STATE $t_{dec} \leftarrow t_{dec}+\tau_{\mathbf{S^m}}(t_{dec})$
\STATE $m \leftarrow m+1$
\ENDWHILE
	\end{algorithmic}
\end{algorithm} 
 \section{Simulation Results} \label{sec:simulations}
The goal of this section is to evaluate the performance of the proposed scheduling algorithms in comparison to the multi-cell transmission with no scheduling and a previously proposed penalty based scheduling algorithm for successive transmission in a single cell network. The multi-cell transmission with no scheduling, denoted by MCNS, randomly selects a single user from each cell for the concurrent transmission by using the continuous rate within a transmission slot. The previously proposed penalty based scheduling algorithm in \cite{MLSP}, denoted by MPA, aims to minimize the schedule length of the users in a single cell full-duplex WPCN by scheduling the users based on their individual penalty. For our simulations, MPA picks the user with minimum penalty among all the cells and allocates the current time slot to this user. 

Simulation results are obtained by averaging $1000$ independent random network realizations. The cells are non-overlapping and distributed within a circle of radius $100$m. The users are uniformly distributed within each cell with radius $10$m. The attenuation of the links considering large-scale statistics are determined by using the path loss model given by 
$PL(d)=PL(d_0)+10\alpha log_{10}\bigg(\frac{d}{d_0}\bigg)+\emph{Z}$,
where $PL(d)$ is the path loss at distance $d$, $d_0$ is the reference distance, $\alpha$ is the path loss exponent, and $Z$ is a zero-mean Gaussian random variable with standard deviation $\sigma$. The small-scale fading has been modeled by using Rayleigh fading with scale parameter $\Omega$ set to mean power level obtained from the large-scale path loss model. The parameters used in the simulations are $D_i=100$ bits for $i \in \{1,\dots, N\}$; the constant rate $r=50$ Kbps; $W= 1$ MHz; $d_0=1$ m; $PL(d_0)=30$ dB; $\alpha=2.7$, $\sigma=4$ \cite{harvest_07, harvest_04,harvest_50}. The self interference coefficient $\beta$ is $-70$ dBm, the initial battery level of the users are $10^{-9}$ J, $P_{max}=0.1mW$ and $P_h=1W$ for the simulations, unless otherwise stated \cite{MLSP}.
 \begin{figure}[t]
 \centering
\includegraphics[width= 0.5 \linewidth]{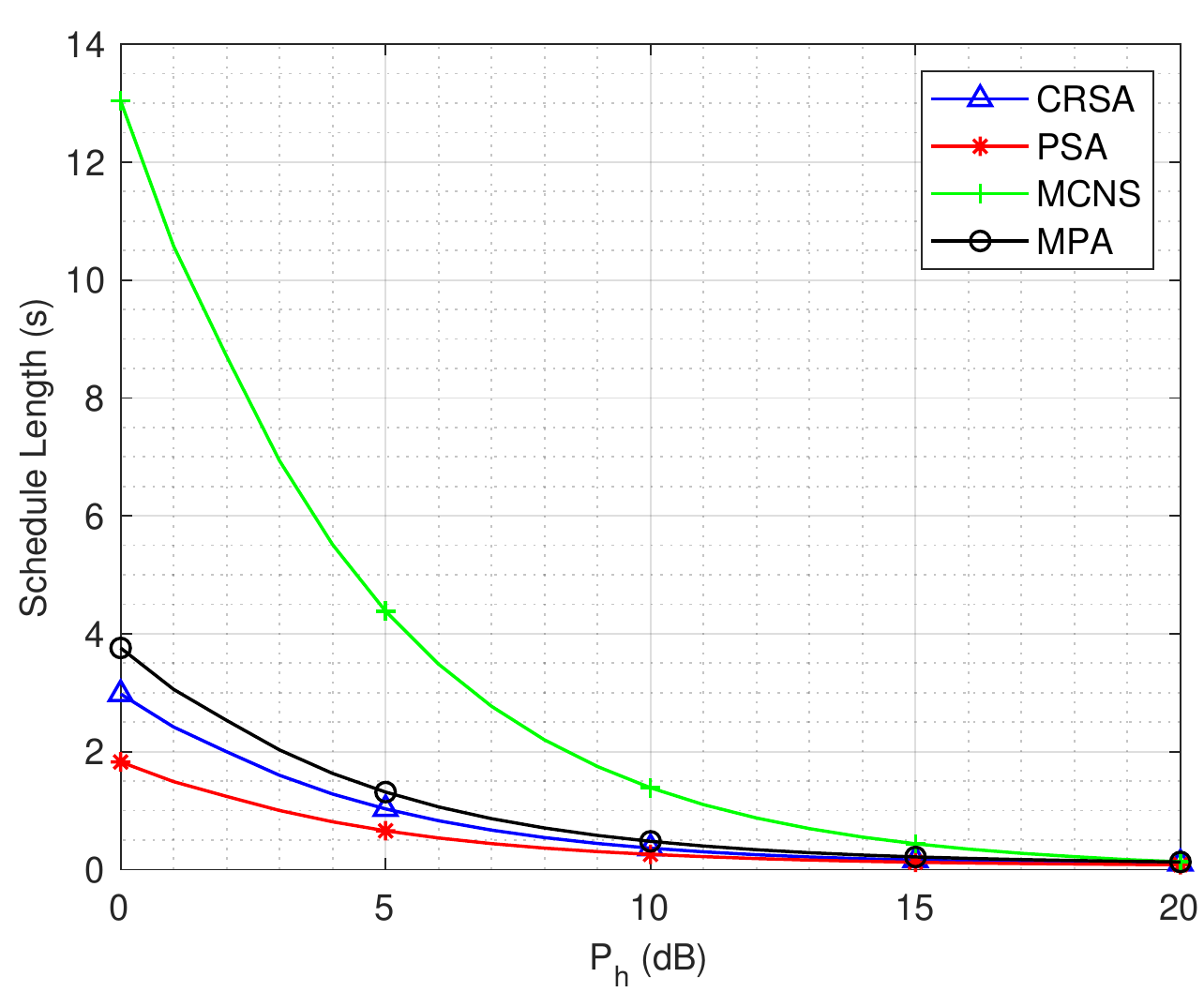}
\caption{Schedule length of the algorithms for different HAP transmit powers in a network of $10$ cells and $50$ users.} \label{Figure:Ph}
\end{figure}
Fig. \ref{Figure:Ph} shows the schedule length of the algorithms for different HAP transmit powers in a network of $10$ cells and $50$ users. The schedule length decreases as the HAP transmit power increases, since higher HAP power allows the users to harvest more energy, which allows users to reduce the waiting time and transmit information at higher rate. For lower values of HAP power, MCNS results in highest schedule length, since MCNS groups the users randomly without considering their penalties, energy harvesting rates and interference levels. As the transmission slot length should be equal for all the concurrently transmitting users, it is more probable that each group may contain a user with very low energy harvesting rate or severe interference, resulting in a higher slot length. For CRSA at lower $P_h$ values, due to low energy harvesting rate, the users need to wait longer to achieve the constant transmission rate, increasing the schedule length. However, proper scheduling of the users for concurrent transmission gives an advantage to CRSA over MCNS and MPA in terms of the schedule length. The MPA has higher schedule length than CRSA, since in MPA, only single user can transmit in each time slot, which results in a higher schedule length. However, proper scheduling of the users in each time slot provides and advantage to MPA over the MCNS. The proposed penalty based multi-cell scheduling algorithm, i.e., PSA, significantly outperforms the MCNS, MPA and CRSA for a practical range of $P_h$ values due to proper grouping, and then, scheduling of the concurrently transmitting users. For large values of the $P_h$, i.e., around $20dB$, all the algorithms perform similar, since users can quickly reach the $P_{max}$ value, which removes the necessity of scheduling.  
 
  \begin{figure}[t]
 \centering
\includegraphics[width= 0.5 \linewidth]{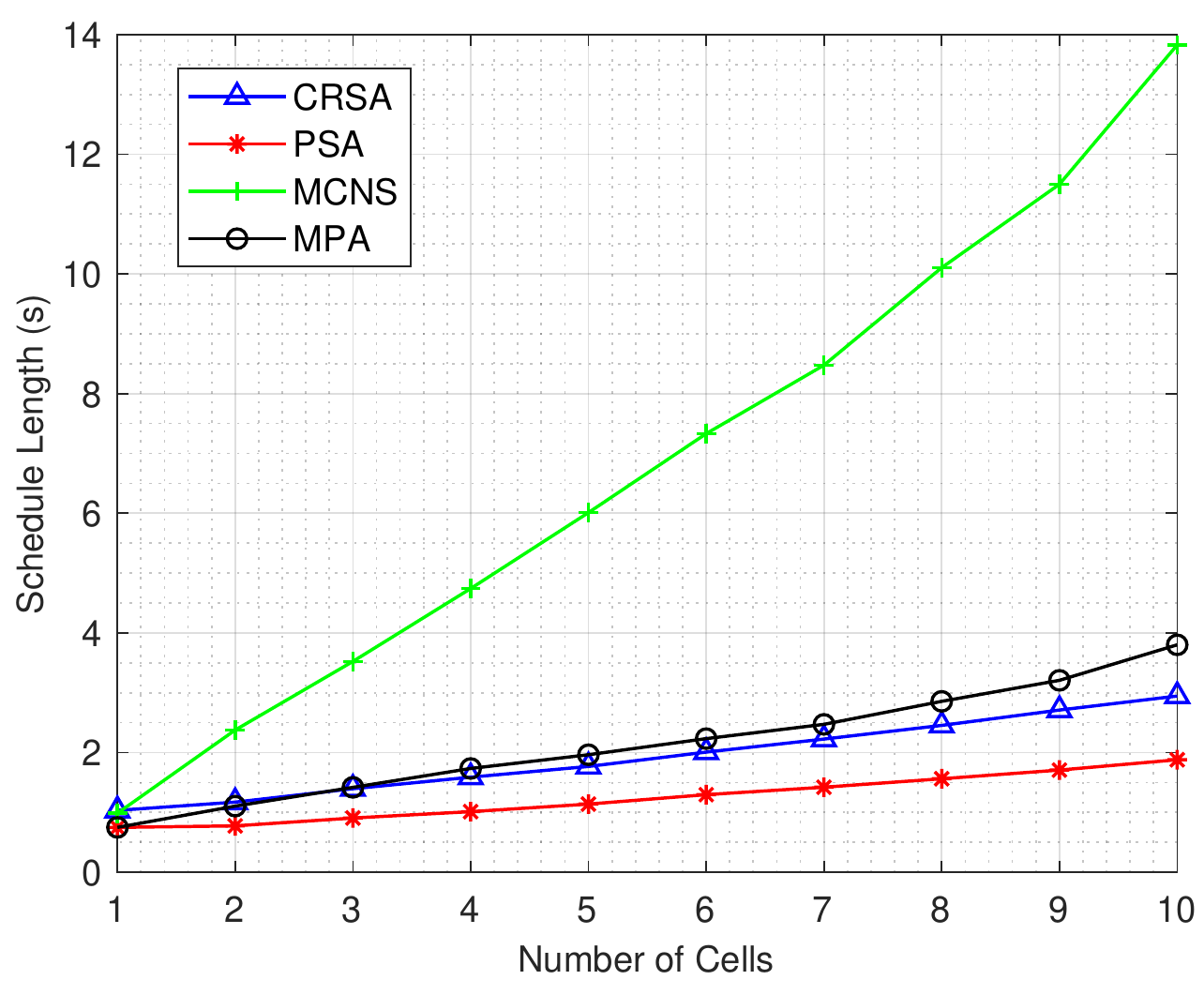}
\caption{Schedule length of the algorithms for different number of HAPs with $5$ users in each HAP.} \label{Figure:cells}
\end{figure}

Fig. \ref{Figure:cells} shows the performance of the algorithms for different number of cells in a network with $5$ users in each cell. The schedule length increases as the number of cells increases in the network, since the higher number of cells results in higher number of users, which requires more transmission combinations due to different interference and energy harvesting rates of the users. For a single cell network, MPA and PSA perform similar. This is due to fact that there is no concurrent transmissions in a single cell network and both algorithms aim to minimize the sum of penalties of the users by picking the user for each slot that has minimum penalty among the unscheduled users. On the other hand, due to constant rate of CRSA, users need to wait to achieve the desired SNR and no scheduling mechanism of MCNS results in higher schedule length. The proper scheduling of the users in CRSA and PSA reduces the schedule length in comparison to MCNS and MPA for a network containing a higher number of cells. For the MCNS, the steep increase in the schedule length is due to the concurrent allocation of higher number of unmatched users. For MCNS, CRSA and MPA algorithms, the addition of new cell results in almost constant increase, since the addition of new cell results in more users in the network, hence, more waiting time and additional interference for the concurrent transmission. On the other hand, for PSA, the addition of new cell results in a diminishing increase in the schedule length, since the chances of combining the similar interference characteristics users increases or a user from the new cell can be accommodated within the same transmission slot without any increase in the transmission slot length.   
  \begin{figure}[t]
 \centering
\includegraphics[width= 0.5 \linewidth]{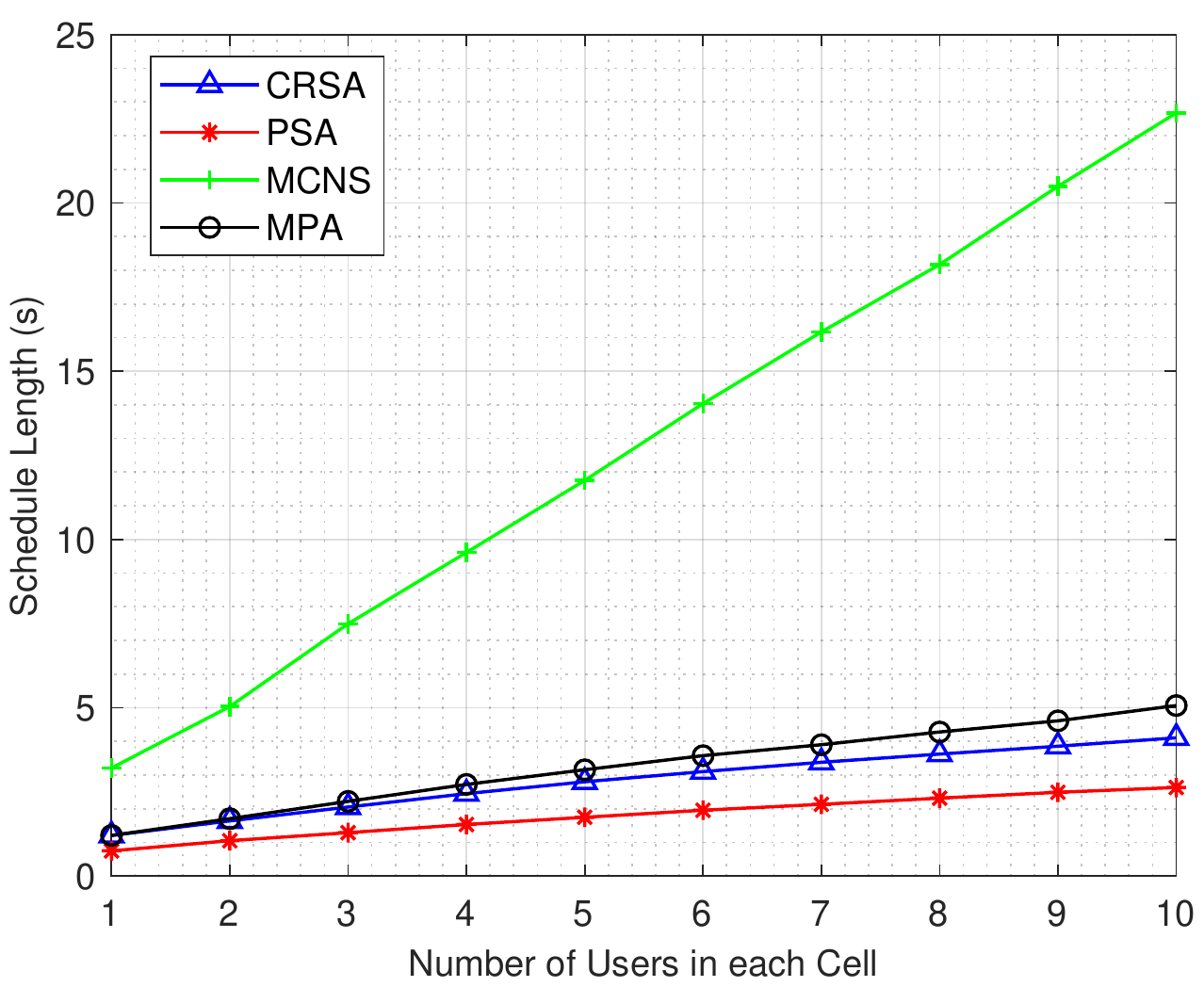}
\caption{Schedule length of the algorithms for different number of users in each cell in a network of $5$ HAPs.} \label{Figure:N}
\end{figure}

Fig. \ref{Figure:N} illustrates the impact of network size on the schedule length for the algorithms. The schedule length increases as the number of users increases within each cell, since this increase requires more transmission slots to complete the transmission of all users. The schedule lengths of CRSA, MPA and MCNS increase almost linearly as the number of users increases in a cell. On the other hand, the effect of each new user on the PSA algorithm is diminishing, since higher number of users within each cell, increases the probability of finding a more suitable user for the concurrent transmission, resulting in smaller schedule length. 
%

\section{Conclusion} \label{sec:conclusion}

In this paper, we consider a full-duplex multi-cell wireless powered communication network, in which multiple hybrid access points transfer RF energy and users harvest this energy to transmit their information. The transmission frame is divided into multiple transmission slots of variable length and in each transmission slot, multiple users from different cells transmit their information concurrently. We investigate the minimum length scheduling problem to determine the power control and scheduling while considering the traffic demand, maximum transmit power and energy causality of the users for constant and continuous rate transmission models. We formulate two mixed integer programming problems, which are difficult to solve for the global optimal solution. As a solution strategy, we decompose the optimization problems into power control and scheduling problems. First, we solve the power control problems based on the evaluation of Perron-Ferobenius conditions and bisection method based algorithm for constant and continuous rate models, respectively. Then, the proposed optimal power control solutions are used to determine the optimal transmission time for a subset of users that will be scheduled by the scheduling algorithms. For the scheduling problem of constant rate model, we propose a heuristic algorithm based on the maximization of the number of concurrently transmitting users by maximizing the allowed interference on each user. For the scheduling problem of continuous rate model, first, we define a penalty function as a metric representing the advantage of concurrent transmission over individual transmission of those users. Then, we show that the schedule length minimization is equivalent to the minimization of the sum of penalties. By using this definition, we propose an algorithm aiming at minimizing the sum of penalties over the schedule. Through extensive simulations, we demonstrate that the proposed solutions outperform the conventional successive transmission and concurrent transmission of randomly selected users for different HAP transmit powers, network densities and network size.    

\bibliography{bib_shahid}
\bibliographystyle{ieeetr}
\end{document}